\documentclass[a4paper,11pt,reqno]{amsart}
\usepackage{amsmath, amsthm, amsfonts, amssymb, mathtools}
\usepackage[margin=2.5cm]{geometry}
\usepackage[english]{babel}
\usepackage{hyperref}
\newtheorem{theorem}{Theorem}
\newtheorem{lemma}[theorem]{Lemma}
\theoremstyle{definition}
\newtheorem{definition}[theorem]{Definition}
\newcommand{\N}{\mathbb{N}}
\newcommand{\R}{\mathbb{R}}
\newcommand{\dash}{\nobreakdash-\hspace{0pt}}
\DeclarePairedDelimiter{\abs}{\lvert}{\rvert}
\DeclarePairedDelimiter{\paren}{(}{)}
\DeclarePairedDelimiter{\brackets}{[}{]}
\newcommand{\eval}[3][]{\brackets[#1]{#2}_{#3}}
\newcommand{\size}[2][]{\abs[#1]{#2}}
\newcommand{\closedball}[3][]{B\brackets[#1]{#3;#2}}

\begin{document}

\title{A more reasonable proof of Cobham's theorem}
\author{Thijmen J. P. Krebs}
\address{
  Delft Institute of Applied Mathematics \\
  Delft University of Technology \\
  PO Box 5031 \\
  2600 GA Delft \\ 
  Netherlands}
\email{\href{mailto:t.j.p.krebs@protonmail.com}{\nolinkurl{t.j.p.krebs@protonmail.com}}}

\begin{abstract}
We present a short new proof of Cobham's theorem without using Kronecker's approximation theorem,
making it suitable for generalization beyond automatic sequences.
\end{abstract}

\maketitle

\section{Introduction}

In this note we give a short proof of the following celebrated theorem on automatic sequences.

\begin{theorem} [Cobham] \label{theorem:cobham}
Let $a,b \in \N_{\geq 2}$ be multiplicatively independent (i.e.~$a^m\neq b^n$ for all $m,n\in \N_{>0}$). A sequence $(f_x)_{x\in\N}$ is $a$\dash{} and $b$\dash{}automatic if and only if it is ultimately periodic.
\end{theorem}

The theorem is originally proven in~\cite{cobham}. To quote~\cite[p.~118]{eilenberg}:
``The proof is correct, long and hard. It is a challenge to find a more reasonable proof of this fine theorem.''.
In response, \cite{hansel-cobham} suggested an easier approach that starts by showing the sequence is syndetic (i.e.~the gap between any successive occurrences of any value is bounded) before proving it is ultimately periodic with a combinatorial argument.
See for instance the proof in~\cite{allouche-book} (together with~\cite{rigo-errata}).

Syndeticity is typically established using Kronecker's approximation theorem (see~\cite{durand-survey-cobham}),
but that proved to be problematic for automatic functions on the Gaussian integers as shown in~\cite{hansel-gaussian}.
Our proof differs entirely from the classical approach and only needs a consequence of the much weaker approximation theorem of Dirichlet.

\section{Preliminaries}

We assume basic familiarity with formal language terminology, and briefly recall a few standard notions from automatic sequence theory.
We refer to~\cite{allouche-book} for a comprehensive treatment.

\begin{definition}
A \emph{deterministic finite automaton with output} (\emph{DFAO}) is a tuple $(S,D,\delta,s_0,F)$, where $S$ is a finite set of \emph{states}, $D$ a finite \emph{input alphabet}, $\delta\colon S\times D \to S$ a \emph{transition function}, $s_0 \in S$ an \emph{initial state}, and $F$ an \emph{output function} on $S$.
On input $w \in D^*$ it outputs $F(\delta(s_0,  w))$, where we extend $\delta(s_0, w) = \delta(\delta(s_0,u),v)$ for any $u,v\in D^*$ with $w=uv$ as usual.
\end{definition}

\begin{definition}
In base $b \in \N_{\geq 2}$, a word $w \in \N^*$ of length $n$ represents the natural number
$\eval{w}{b} = w_0b^{n-1} + \ldots + w_{n-2}b + w_{n-1}$, and a language $L \subseteq \N^*$ represents $\eval{L}{b} = \{\eval{w}{b}\mid w \in L\}$.
\end{definition}

\begin{definition}
Let $b \in \N_{\geq 2}$ and $\{0,1,\ldots,b-1\} \subseteq D \subseteq \N$ be finite.
A sequence $(f_x)_{x\in\N}$ is $(b,D)$\dash{}\emph{automatic} if there is a DFAO $(S,D,\delta,s_0,F)$ such that $f_{\eval{w}{b}} = F(\delta(s_0, w))$ for all $w \in D^*$.
A sequence $(f_x)_{x\in\N}$ is $b$\dash{}\emph{automatic} if it is $(b,\{0,1,\ldots,b-1\})$\dash{}automatic.
\end{definition}

\begin{lemma} \label{lemma:digits-auto}
Let $b \in \N_{\geq 2}$ and $\{0,1,\ldots,b-1\} \subseteq D \subseteq \N$ be finite.
A sequence $(f_x)_{x\in\N}$ is $b$\dash{}automatic if and only if it is $(b,D)$\dash{}automatic.
\end{lemma}
\begin{proof}
Adapt~\cite[Thm.~6.8.6]{allouche-book} to use the transducer of~\cite[Prop.~7.1.4]{lothaire} for normalization on~$D^*$.
\end{proof}

Any two bases have relatively close powers, which follows easily from Dirichlet's approximation theorem or by mimicking its proof to avoid logarithms as follows.

\begin{lemma} \label{lemma:base-approximation}
Let $a,b \in \N_{\geq 2}$ and $\epsilon \in \R_{>0}$. Then there are $m,n \in \N_{>0}$ such that $\abs{a^m-b^n} \leq \epsilon b^n$.
\end{lemma}
\begin{proof}
We may assume that $a\geq b$ by taking a suitable power of $a$, so the sequence $(f_x)_{x\in \N}$ given by $a^xb^{-f_x} \in [1,b)$ for all $x \in \N$ is strictly increasing.
By the pigeonhole principle there are natural numbers $x < y$ such that $\abs{a^{y}b^{-f_y} - a^{x}b^{-f_x}} \leq \epsilon$,
that is, $\abs{a^{y-x} - b^{f_y-f_x}} \leq \epsilon b^{f_y}a^{-x} \leq \epsilon b^{f_y-f_x}$.
\end{proof}

A sequence $(f_x)_{x\in\N}$ has \emph{local period} $p \in \N_{>0}$ on an interval $I \subseteq \N$ if $f_x = f_{x+p}$ for all $x,x+p \in I$.
Local periodicity on sufficiently overlapping intervals extends to their union.

\begin{lemma} \label{lemma:overlapping-local-period}
Let $(f_x)_{x\in\N}$ have local period $p$ on an interval $I$ and local period $q$ on an interval $J$.
If $\size{I \cap J} \geq p+q$, then $f$ has local period $p$ on the interval $I \cup J$.
\end{lemma}
\begin{proof}
Pick any $x,x+p \in I \cup J$.
If $x,x+p \in I$, we have $f_x = f_{x+p}$ by assumption. 
Otherwise, since the interval $I\cap J$ has cardinality at least $p+q$, we have $x,x+p \in J$ and we can pick $y,y+p \in I \cap J$ such that $y \equiv x \pmod q$.
Then $f_x = f_y = f_{y+p} = f_{x+p}$ by local periodicity on $J$, $I$ and $J$ respectively.
\end{proof}

\section{Proof}

Let $\closedball{r}{x} = \{y \in \N \mid \abs{y-x} \leq r\}$ be the interval centered on $x \in \R_{\geq 0}$ with radius $r \in [0,x]$.

\begin{proof}[Proof of theorem~\ref{theorem:cobham}.]
As usual, we only prove the forward direction.

For each $c \in \{a,b\}$, $f$ is computed by a DFAO $(S_c,D_c,\delta_c,s_{0,c},F_c)$ in base $c$ with digits $D_c = \closedball{c}{c}$ by lemma~\ref{lemma:digits-auto}. 
It is easy to check that $\closedball{c^n}{c^n} \subseteq \eval{D_c^n}{c}$ for all $n \in \N_{>0}$.
Define $L_{cs} = \{w \in D_c^* \mid \delta_c(s_{0,c}, w) = s\}$ for $s \in S_c$. Then for all $w \in L_{cs}$ and $v \in D_c^*$ we have
$f_{\eval{wv}{c}} = F_c(\delta_c(s_{0,c},wv)) = F_c(\delta_c(s, v))$, so for all $x,y \in \eval{L_{cs}}{c}$, $n \in \N$ and $z \in \eval{D_c^n}{c}$
\begin{equation} \label{eq:teleport-auto}
f_{xc^n+z} = f_{yc^n+z}.
\end{equation}

We create local periods of $f$ as follows.
Let $S_{\infty}$ be the set of $s \in S_b$ for which $\eval{L_{bs}}{b}$ is infinite.
Since $\{\eval{L_{at}}{a} \mid t \in S_a\}$ is a finite cover of $\N$, 
we can fix for each $s \in S_{\infty}$ some $t \in S_a$ and distinct $x_{st}, y_{st} \in \eval{L_{bs}}{b} \cap \eval{L_{at}}{a}$.
Letting $\xi = \max\{x_{st},y_{st} \mid s \in S_{\infty}\}+1$, we can find $m,n \in \N_{>0}$ such that $\xi \abs{a^m-b^n} \leq \frac{1}{6} b^n$ by lemma~\ref{lemma:base-approximation}.
In particular, we get $\frac{5}{6} b^n \leq a^m$.
Since $a^m \neq b^n$ we can take $p_{st} = (x_{st}-y_{st}) (a^m-b^n) \in (0,\frac{1}{6}b^n]$ for all $s\in S_{\infty}$ by swapping $x_{st}$ and $y_{st}$ if necessary.

We show for each $s \in S_{\infty}$ and $x \in \eval{L_{bs}}{b}$ that $f$ has local period $p_{st}$ on the interval $I_x = \closedball{\frac{2}{3}b^n}{xb^n+b^n}$. 
Pick any $z,z+p_{st} \in \closedball{\frac{2}{3}b^n}{b^n} \subseteq \eval{D_b^n}{b}$. 
Since
\[
\abs[]{z-y_{st}(a^m-b^n) - a^m}
\leq \abs[]{z-b^n} + (y_{st}+1) \abs[]{a^m-b^n}
\leq \tfrac{5}{6}b^n \leq a^m,
\]
we have $z-y_{st}(a^m-b^n) \in \closedball{a^m}{a^m} \subseteq \eval{D_a^m}{a}$.
Hence, using~\eqref{eq:teleport-auto} thrice we see as desired
\[
\begin{aligned}
f_{xb^n+z} 
&= f_{y_{st}b^n+z} \\
&= f_{y_{st}a^m + z - y_{st}(a^m-b^n)} \\
&= f_{x_{st}a^m + z - y_{st}(a^m-b^n)} \\
&= f_{x_{st}b^n + z + p_{st}} \\
&= f_{xb^n + z + p_{st}}.
\end{aligned}
\]

Let $x \in \N$ be such that $\{\eval{L_{bs}}{b}\mid s \in S_{\infty}\}$ covers $x+\N$, and fix for $f$ a local period $p_y \leq \frac{1}{6}b^n$ on $I_y$ for all $y \geq x$.
We show that $f$ has local period $p_x$ on $\bigcup_{x\leq y\leq z} I_y$ for all $z \geq x$ by induction.
It surely holds if $z = x$. 
Otherwise, $f$ has local period $p_x$ on $\bigcup_{x\leq y < z} I_y$ by induction and local period $p_z$ on $I_z$, so
lemma~\ref{lemma:overlapping-local-period} proves our induction hypothesis as
$\paren[\big]{\bigcup_{x\leq y < z} I_y} \cap I_z = \closedball{\frac{1}{6}b^n}{(z+\frac{1}{2})b^n}$ has cardinality at least $\lfloor \frac{1}{3}b^n\rfloor \geq 2\lfloor \frac{1}{6}b^n\rfloor \geq p_x+p_z$.

We conclude that $f$ has local period $p_x$ on $\bigcup_{x \leq y} I_y$, that is, $f$ is ultimately periodic.
\end{proof}

\section{Future work}

We will show that our approach extends well to prove the Cobham-Semenov theorem from~\cite{semenov}, and to prove the Cobham-type theorem for automatic functions on vectors of imaginary quadratic integers. In particular, we will establish the conjecture for the Gaussian integers from~\cite{hansel-gaussian}.

\bibliographystyle{abbrv}
\bibliography{cobham-article}

\end{document}